\documentclass{article}
\usepackage{lipsum}
\usepackage{tikz}
\usepackage{xifthen}
\usepackage{listings}
\RequirePackage{amsthm,amsmath,amssymb}
\usepackage{natbib}

\newcommand{\doBlank}[1]{}
{}

\def\?#1{}

\newcommand{\dif}{\mathrm{d}}

\numberwithin{equation}{section}
\theoremstyle{plain}

\newtheorem{assumption}{Assumption}
\newtheorem{lemma}{Lemma}[section]
\newtheorem{definition}{Definition}
\newtheorem{proposition}{Proposition}[section]
\newtheorem{corollary}{Corollary}[section]

\newtheorem{theorem}{Theorem}[section]

\newcommand{\writetitle}{0}
\newcommand{\mytitle}[1]
{   \ifthenelse{\writetitle=1}{}{}
}

\newread\mysource
\begin{document}
\title{Rate optimality of random-walk Metropolis algorithm  in high-dimension with heavy-tailed target distribution}
\date{}
\author{Kengo KAMATANI\footnote{Supported in part by Grant-in-Aid for Young Scientists (B) 24740062.}}

\maketitle
\begin{abstract} 
The choice of the increment distribution is crucial for the random-walk Metropolis-Hastings (RWM) algorithm. In this paper we study the optimal choice in high-dimension setting among all possible increment distributions.  
The conclusion is rather counter intuitive, but the optimal rate of convergence is attained by the usual choice, the normal distribution as the increment distribution. In particular, no heavy-tailed increment distribution can improve the rate. 
\end{abstract}

\section{Introduction}
Markov chain Monte Carlo (MCMC) methods are widely used techniques for evaluation of complicated integrals. 
The random-walk Metropolis-Hastings (RWM) algorithm is one of the major subclass of MCMC methods. Given a current state $x\in\mathbb{R}^d$, 
RWM algorithm propose a new value 
\begin{align*}
	x^*\leftarrow x+w,\ w\sim\Gamma^d
\end{align*}
where  increment distribution $\Gamma^d$ is a probability measure symmetric about the origin, that is, 
$\Gamma^d(A)=\Gamma^d(-A)$ for any Borel set $A$ of $\mathbb{R}^d$.  
The proposed value $x^*$ will be accepted with probability 
\begin{align*}
	\alpha(x,x^*)=\min\left\{1,\frac{p^d(x^*)}{p^d(x)}\right\},
\end{align*}
where $P^d(\dif x)=p^d(x)\dif x$ is the target probability distribution.

The appeal of the RWM algorithm is its generality. 
Virtually any kinds of  increment distribution can be used for this algorithm
as long as it is symmetric about the origin.  Even if we decide to use the normal distribution, 
we still have a free choice of its covariance structure. 

This choice of increment distribution is a crucial part of Bayesian analysis. 
There are many choices: (a) light-tailed distribution: ex. normal distribution, truncated distribution (b) heavy-tailed distribution: ex. student distribution, stable distribution (c) mixed case: ex. update only one component at each iteration. For (a), for the normal distribution, a useful criterion was proposed by \cite{MR1428751}. On the other hand, for heavy-tailed target distributions,   \cite{MR2396939} proved the benefit of using (b). However the optimal choice of all possible increment distributions is not studied yet, though its choice has significant effect of the performance. 

In this paper we consider optimal choice of increment distribution in terms of the convergence rate as $d\rightarrow\infty$. 
We will obtain the following type of results. 

\begin{theorem}
Suppose that $P^d$ is the $d$-dimensional standard normal distribution, and $\Gamma^d$ is symmetric about the origin. 
Let $X_m^d=(X_{m,1}^d,\ldots, X_{m,d}^d)\in\mathbb{R}^d\ (m=0,1,\ldots)$ be the output of the RWM algorithm and let $X_0^d\sim P^d$. Then for any $\alpha(d)=o(d)$, 
there exists $k_d\in \{1,\ldots, d\}$ such that
\begin{align*}
	\sup_{0\le i,j\le \alpha(d)}|X^d_{i,k_d}-X^d_{j,k_d}|\overset{p}{\to}0\ (d\rightarrow\infty). 
\end{align*}
\end{theorem}

\noindent
Therefore, if the number of iteration $\alpha(d)$ is shorter than $d$, the output of RWM algorithm is degenerate at least one coordinate for any choices of increment distributions. On the other hand, as studied in \cite{MR1425429}, 
there is a non-degenerate limit if $\alpha(d)=d$ and $\Gamma^d$ is the normal distribution. Therefore $\alpha(d)=d$ is the optimal rate of iteration to obtain a non-degenerate limit. 
This rate can not be improved for any choice of $\Gamma^d$. 
Moreover, we will show that, this is true even for heavy-tailed case: for heavy-tailed target distribution, the optimal rate becomes $d^2$  and this rate is again attained by the normal increment distribution.  Therefore, the overall conclusion is that for both cases, usual choice of the increment distribution  attains the optimal rate, which is rather disappointing fact. Note that there exists a strategy which improves the rate other than RWM. See \cite{arXiv:1412.6231} for the detail.

We use the following notation throughout in this paper. 
The state space is $\mathbb{R}^d$ throughout, and the Euclidean norm is denoted   $\|\cdot\|$ and 
the inner product is denoted $\langle\cdot,\cdot\rangle$. 
Write $N_d(\mu,\Sigma)$ for $d$-dimensional normal distribution
with the mean vector $\mu\in\mathbb{R}^d$ and the variance covariance matrix $\Sigma$, 
and $\phi_d(x;\mu,\Sigma)$ be its probability distribution function. 
The $d\times d$-identity matrix is denoted by $I_d$. 
Write $\phi(z)=\exp(-z^2/2)/\sqrt{2\pi}$ and  
$\Phi(z)=\int_{-\infty}^z\phi(w)\dif w$.

We also use the notation $\mathcal{L}(X)$ to denote the law of a random variable $X$. For $x\in\mathbb{R}$, write $x^+=\max\{0,x\}$ and $x^-=\max\{0,-x\}$. 
For $x\ge 0$,  $\left[x\right]$ is the integer part of $x\ge 0$, 
and $K_a=(a^{-1},a)$ is an open interval in $\mathbb{R}_+=(0,\infty)$ for $a>1$. 
Write the sup norm by $\|h\|_\infty=\sup_{s\in E}|h(s)|$ for $h:E\rightarrow\mathbb{R}$ for a state space $E$.  
If $h$ is absolutely continuous, we write $h'$ for the derivative.

\section{Preliminary}\label{mainsec}
\subsection{Assumption for the target distribution}
We consider a sequence of the target distributions $(P^d)_{d\in\mathbb{N}}$ indexed by the number of dimension $d$. 
For a given $d$, $P^d$ is a $d$-dimensional probability distribution that is a scale mixture of the normal distribution.  
Furthermore, our asymptotic setting is that the number of dimension $d$ goes infinity while the mixing distribution $Q$ is unchanged. 
  
Let $Q(\dif y)$ be a probability measure on $[0,\infty)$. 
Let $P^d$ be the scale mixture of the normal distribution defined by 
\begin{equation}\label{mainseceq}
P^d=\mathcal{L}(X^d_0),\ Q^d=\mathcal{L}(\|X^d_0\|^2/d)
\end{equation}
where $X^d_0|Y\sim N_d(0,YI_d)$ and $Y\sim Q$. 
By this assumption, in particular, $P^d$ is rotationally symmetric, that is, it is invariant under all orthogonal transform. 
In this paper, we call that $P^d$ is light-tailed if $Q$ is a discrete distribution, and 
we call that $P^d$ is heavy-tailed if $Q$ has a density with respect to the Lebesgue measure. 
Moreover, we assume that the mixing distribution $Q$ is $\delta_1$ (the Dirac measure charging $1\in\mathbb{R}$) or it satisfies the following. 
\begin{assumption}\label{assumption-1}
Probability distribution $Q$ has the strictly positive probability distribution function $q(y)$. 
The probability distribution function $q(y)$ is  continuously differentiable
and $q(y)$ and $q'(y)$ vanish at $+0$ and $+\infty$. 
Moreover, $\lim_{y\rightarrow +\infty}yq(y)=0$. 
\end{assumption}

Under this assumption, \cite{2016K} showed the following. 

\begin{lemma}\label{elementary}
Under Assumption \ref{assumption-1}, both $P^d$ and $Q^d$ have the probability distribution functions $p^d$
and $q^d$ that satisfy
\begin{equation}\nonumber
p^d(x)\propto \|x\|^{2-d} q^d\left(\frac{\|x\|^2}{d}\right),
\end{equation} 
for $x\neq 0$. Moreover,  $p^d(x_1)<p^d(x_2)$ if and only if $\|x_1\|>\|x_2\|$. 
\end{lemma}

%
%
%

\subsection{Consistency for high dimensional MCMC}
In this section, we review consistency of MCMC studied  in \cite{Kamatani10}. 
Set a sequence of Markov chains $\{\xi^d:=(\xi^d_m;m\in\mathbb{N}_0)\}\ (d\in\mathbb{N})$ with the invariant probability measures $\{\Pi^d\}_d$. 
The sequence $\{\xi^d\}_d$ is called consistent if 
\begin{equation}\label{consistency}
\frac{1}{M}\sum_{m=0}^{M-1}f(\xi^d_m)-\Pi^d(f)=o_{\mathbb{P}}(1)
\end{equation}
for any $M, d\rightarrow\infty$ for any bounded continuous function $f$. This  says that 
the integral $\Pi^d(f)=\int f(x)\Pi^d(\dif x)$ we want to calculate is approximated by Monte Carlo simulated value
$\frac{1}{M}\sum_{m=0}^{M-1}f(\xi^d_m)$ after a reasonable number of iteration $M$.  
Regular Gibbs sampler should satisfy this type of property (more precisely, local consistency. See \cite{Kamatani10})
when $d$ is the sample size of the data. However this is not always the case as described in the last part of \cite{Kamatani10}. 
In our case, (\ref{consistency}) is not satisfied in two respects: 
the state space is not the same for $d\in\mathbb{N}$, and 
$M=M(d)$ should satisfy a certain rate.

For $d=1,2,\ldots$,  let $X^d=(X^d_m;m\in\mathbb{N}_0)$ be an $\mathbb{R}^d$-valued stationary process
 with invariant distribution $P^d$ on a discrete-time stochastic basis $(\Omega^d,\mathcal{F}^d,\mathbf{F}^d,\mathbb{P}^d)$.  
We consider asymptotic properties of $X^d$ through $d\rightarrow\infty$.  
As commented above, the state space for  $X^d\ (d\in\mathbb{N})$
 changes as $d\rightarrow\infty$ that is inconvenient for further analysis.  
 To overcome the difficulty, we set a projection $\pi_E=\pi^d_E$ for a finite subset $E\subset\{1,\ldots, d\}$ by
 \begin{equation}\nonumber
 \pi_E(x)=(x_i)_{i\in E}\ (x=(x_i)_{i=1,\ldots, d}). 
 \end{equation}
 We denote $\pi_k$ for $\pi_{\{1,\ldots, k\}}$. 
 
 \begin{definition}[Consistency]
A sequence of $\mathbb{R}^d$-valued process $\{X^d\}_{d\in\mathbb{N}}$ is weakly consistent with the rate $\alpha(d)$ if 
 \begin{equation}\label{eq1}
\frac{1}{M(d)}\sum_{m=0}^{M(d)-1}f\circ\pi_{E^d_k}(X^d_m)-P^d(f\circ\pi_{E^d_k})=o_{\mathbb{P}}(1)
 \end{equation}
as $d\rightarrow\infty$ for any 
$k\in\mathbb{N}$, $M(d)\rightarrow\infty$ such that $\alpha(d)/M(d)=o(1)$ and for any bounded continuous function $f:\mathbb{R}^k\rightarrow\mathbb{R}$
and any $k$-elements $E^d_k$ of  $\{1,\ldots, d\}$. 
The random-walk Metropolis algorithm is weakly consistent with the rate $\alpha(d)$ if above $X^d$ is generated by the algorithm. 
 \end{definition}

\section{Main results}\label{main}

For the target probability distribution $P^d(\dif x)=p^d(x)\dif x$ on $\mathbb{R}^d$, the random-walk Metropolis algorithm
generates Markov chain defined by 
\begin{align}\label{eq:mh}
	X_m^d=\left\{\begin{array}{ll}X_{m-1}^d+W_m^d&\mathrm{with\ probability}\ \min\left\{1,\frac{p^d(X_{m-1}^d+W_m^d)}{p^d(X_{m-1}^d)}\right\}\\X_{m-1}^d&\mathrm{with\ probability}\ 1-\min\left\{1,\frac{p^d(X_{m-1}^d+W_m^d)}{p^d(X_{m-1}^d)}\right\}
	\end{array}\right.
\end{align}
for $m\ge 1$, where $(W_m^d)_m$ are independent and identically distributed random variables from a probability measure $\Gamma^d$. The probability measure $\Gamma^d$ should be symmetric about the origin, that is, $\Gamma^d(A)=\Gamma^d(-A)$ where $-A=\{x\in\mathbb{R}^d; -x\in A\}$. For simplicity, throughout in this paper, we assume $(X_m^d)_m$ is a stationary process, that is, $X_0^d\sim P^d$. 

\subsection{Optimality for Gaussian target distribution}

In this section we will study the optimality property for the Gaussian target distribution case $P^d=N_d(0,I_d)$. 
The random-walk Metropolis generates a Markov chain $(X_m^d)$
defined by 
\begin{align}\label{eq:mh_light}
	X_m^d=\left\{\begin{array}{ll}X_{m-1}^d+W_m^d&\mathrm{with\ probability}\ \alpha(S^d_m)\\X_{m-1}^d&\mathrm{with\ probability}\ 1-\alpha(S^d_m)
	\end{array}\right.
\end{align}
for $m\ge 1$, where $X_0^d\sim P^d$, $W_m^d\sim \Gamma^d$  and 
\begin{align*}
\alpha(s)=\exp(-s^+),\ S^d_m=\langle X_{m-1}^d,W_m^d\rangle+\frac{\|W_m^d\|^2}{2}. 
\end{align*}
Set a filtration $\mathcal{F}_M^d=\sigma\{X_m^d;m\le M\}$. In the following proof, the probability measure $N_\sigma=N(\sigma^2/2,\sigma^2)$ plays an important role. 
Some properties of $N_\sigma$ are summarised in Section \ref{sec:girsanov}.

\begin{proposition}\label{prop:light}
Let $P^d=N_d(0,I_d)$ and $\Gamma^d$ be probability measure on $\mathbb{R}^d$ symmetric about the origin and $\alpha(d)=o(d)$.
Then we can choose $k_d\in \{1,\ldots, d\}$ so that
$\sup_{0\le i,j\le \alpha(d)}\left|X^d_{i,k_d}-X^d_{j,k_d}\right|\overset{p}{\to} 0\ (d\rightarrow\infty)$.
\end{proposition}

\begin{proof}
Set $\xi_m^d=(\xi_{m,1}^d,\ldots, \xi_{m,d}^d)=X_m^d-X_{m-1}^d$. It is sufficient to prove $\sum_{1\le M\le \alpha(d)}|\sum_{m=1}^M \xi_{m,k_d}^d|\overset{p}{\to}0$ for some $k_d\in\{1,\ldots, d\}$. 
We have to prove (\ref{eq:cor:gi}) in Corollary  \ref{cor:gj} for $T^d=\alpha(d)$. 
\begin{enumerate}
\item[(a)]
The first part of (\ref{eq:cor:gi}). Note that by (\ref{eq:mh_light}) we have
\begin{align}\label{eq:prop_light:drift2}
\mathbb{E}\left[\xi_{1,k}^d|\mathcal{F}_0^d\right]&=
\mathbb{E}\left[W_{1,k}^d\alpha(S_1^d)|\mathcal{F}_0^d\right].
\end{align}
Since the proposal distribution $\Gamma^d$ is symmetric about the origin, $(W_{1,k}^d,-W_{1,k}^d)$
is an exchangeable pair conditioned on $\mathcal{F}_0^d$. 
Set $\tilde{S}^d_1=S_1^d-2X_{0,k}^dW_{1,k}^d$ by switching $W_{1,k}^d$ and $-W_{1,k}^d$ of $S^d_1$. 
By Lemma \ref{lem:exchageble} (\ref{eq:lem:exchageble1}),  we deduce that  
\begin{align*}
\mathbb{E}\left[\left|\mathbb{E}\left[\xi_{1,k}^d|\mathcal{F}_0^d\right]\right|\right]&
= \frac{1}{2} \mathbb{E}\left[\left|\mathbb{E}\left[\left.W_{1,k}^d\left(\alpha(S_1^d)-\alpha(\tilde{S}^d_1)\right)\right|\mathcal{F}_0^d\right]\right|\right]\\
&\le \frac{1}{2} \mathbb{E}\left[\left|W_{1,k}^d\left(\alpha(S_1^d)-\alpha(\tilde{S}^d_1)\right)\right|\right].
\end{align*}
Observe that $(S_1^d, \tilde{S}^d_1)$ is also an exchangeable pair conditioned on $W_{1,k}^d$. Thus by Lemma \ref{lem:exchageble} (\ref{eq:lem:exchageble2}), we see that 
\begin{align*}
\frac{1}{2}\mathbb{E}\left[\left.\left|\alpha(S_1^d)-\alpha(\tilde{S}^d_1)\right|\right|W^d_{1,k}\right]
=
2\mathbb{E}\left[\left.\alpha(S_1^d)\left(\frac{1}{2}-\mathbb{P}\left(S_1^d<\tilde{S}^d_1|S_1^d,W_1^d\right)\right)\right|W^d_{1,k}\right]. 
\label{eq43}
\end{align*}
\item[(c)]
Next we estimate the integrand in the above expectation. 
In order to ease the notation, we write  $\sigma=\|W_1^d\|$ and $\rho_k=1-2|W_{1,k}^d|^2/\|W_1^d\|^2$. 
Using this, we have
\begin{align*}
	\mathcal{L}(S_1^d,\tilde{S}_1^d|W_1^d)=N\left(
	\begin{pmatrix}\sigma^2/2\\\sigma^2/2
	\end{pmatrix},
	\begin{pmatrix}\sigma^2&\rho_k\sigma^2\\\rho_k\sigma^2&\sigma^2
	\end{pmatrix}
	\right)
\end{align*}
and hence
\begin{equation}\nonumber
\mathbb{P}(S_1^d<\tilde{S}_1^d|S_1^d, W_1^d)=\Phi\left(-\frac{(1-\rho_k)(S_1^d-\frac{\sigma^2}{2})}{\sigma\sqrt{1-\rho_k^2}}\right).
\end{equation}
Note that $\rho_k\ge 0$ for $k\neq k^*$ where $k^*=\arg\max_{k=1,\ldots, d}|W_k^d|$, that is, $\rho_k$ can be negative at most one coordinate. 
Thus for $k\neq k^*$, we have
\begin{align*}
\left|\frac{1}{2}-\mathbb{P}(S_1^d<\tilde{S}_1^d|S_1^d, W_1^d)\right|
&= 
\left|\Phi(0)-\Phi\left(-\frac{(1-\rho_k)(S_1^d-\frac{\sigma^2}{2})}{\sigma\sqrt{1-\rho_k^2}}\right)\right|\\
&\le\|\phi\|_\infty\frac{(1-\rho_k)\left|S_1^d-\frac{\sigma^2}{2}\right|}{\sigma\sqrt{1-\rho_k^2}}\\
&\le\|\phi\|_\infty\frac{\sqrt{1-\rho_k}
\left|S_1^d-\frac{\sigma^2}{2}\right|}{\sigma}.
\end{align*}
\item[(d)]
By the above estimate
\begin{align*}
\left|\mathbb{E}\left[\xi_{1,k}^d1_{\{k\neq k^*\}}|\mathcal{F}_0^d\right]\right|&
\le 2^{3/2}\|\phi\|_\infty\mathbb{E}\left[\frac{|W_{1,k}^d|^2}{\|W_1^d\|^2}\left|S_1^d-\frac{\|W_1^d\|^2}{2}\right|\alpha(S_1^d)\right]
\end{align*}
and hence
\begin{align}\nonumber
\mathbb{E}\left[\sum_{k=1}^d\left|\mathbb{E}\left[\xi_{1,k}^d1_{\{k\neq k^*\}}|\mathcal{F}_0^d\right]\right|\right]
&\le 
2^{3/2}\|\phi\|_\infty\mathbb{E}\left[\left|S_1^d-\frac{\|W_1^d\|^2}{2}\right|\alpha(S_1^d)\right]\\
&\le 
2^{3/2}\|\phi\|_\infty\left\{\mathbb{E}\left[|S_1^d|\alpha(S_1^d)\right]+2^{-1}\mathbb{E}\left[\|W_1^d\|^2\alpha(S_1^d)\right]\right\}.\label{eq:prop:light1} 
\end{align}
The boundedness of the second term will be proved in (e) and so we omit it. 
By using  $\mathcal{L}(S_1^d|W_1^d)=N(\|W_1^d\|^2/2,\|W_1^d\|^2)$, the first term is bounded since 
\begin{align*}
\mathbb{E}[|S_1^d|\alpha(S_1^d)|W_1^d]=\mu_1(\|W_1^d\|)\le\sup_\sigma\mu_1(\sigma)<\infty
\end{align*}
by Section \ref{sec:girsanov}. 
On the other hand, 
\begin{align}
\mathbb{E}\left[\sum_{k=1}^d\left|\mathbb{E}\left[\xi_{1,k}^d1_{\{k= k^*\}}|\mathcal{F}_0^d\right]\right|\right]
\le\mathbb{E}\left[\left|\xi_{1,k^*}^d\right|\right]
\le\mathbb{E}\left[\|\xi^d_1\|^2\right]\label{eq:prop:light2}, 
\end{align}
and this is bounded by (e) below. 
By (\ref{eq:prop:light1}) and (\ref{eq:prop:light2}), the first part of (\ref{eq:cor:gi}) will be proved by (e). 
\item[(e)] The second part of (\ref{eq:cor:gi}).  
Note that by (\ref{eq:mh_light}), we have
\begin{align*}
\mathbb{E}\left[\|\xi_1^d\|^2\right]=
\mathbb{E}\left[\|W_1^d\|^2\alpha(S_1^d)\right]
=
\mathbb{E}\left[\|W_1^d\|^2\mathbb{E}\left[\left.\alpha(S_1^d)\right|W_1^d\right]\right]
=\mathbb{E}\left[\|W_1^d\|^2\mu_0(\|W_1^d\|)\right]
\end{align*}
where we used $\mathcal{L}(S_1^d|W_1^d)=N(\|W_1^d\|^2/2,\|W_1^d\|^2)$. 
Since $\sup_\sigma\sigma^2\mu_0(\sigma)<\infty$ (see Section \ref{sec:girsanov}), the above is bounded. 
\end{enumerate}
By these estimates, the proof is completed by Corollary \ref{cor:gj}. 
\end{proof}

\begin{theorem}
Let $P^d=N_d(0,I_d)$. 
	If the RWM is weakly consistent with the rate $\alpha(d)$, then $\liminf_{d\rightarrow\infty}\alpha(d)/d>0$. The RWM algorithm has the optimal rate $\alpha(d)=d$ in this sense. 
\end{theorem}

\begin{proof}
If $\liminf_{d\rightarrow\infty}\alpha(d)/d>0$, then by choosing subsequence of $d$, we can assume $\alpha(d)=o(d)$ without loss of generality. 
Choose $M(d)=o(d)$ so that $M(d)/\alpha(d)\rightarrow\infty$. Then there exists $k_d\in \{1,\ldots d\}$ such that 
$\sup_{0\le i,j\le M(d)}|X^d_{i,k_d}-X^d_{j,k_d}|\overset{p}{\to} 0$. Then
\begin{align}\nonumber
\frac{1}{M(d)}\sum_{m=0}^{M(d)-1}f(X_{m,k_d}^d)-Nf=f(X_0)-Nf+o_\mathbb{P}(1)
\end{align}
for any bounded continuous function $f(x)$ where $Nf=\int f(x)\phi(x)\dif x$. The right-hand side of the above can not be $o_\mathbb{P}(1)$ unless $f$ is a constant. 
Thus the RWM does not have weak consistency with the rate $\alpha(d)$ and $E_k^d=\{k_d\}$. On the other hand, by \cite{2016K}, the RWM algorithm has weak consistency with the rate $d$
when $\Gamma^d=N_d(0,l^2I_d/d)$ for $l>0$. Thus $\alpha(d)=d$ is the optimal convergence rate. 
\end{proof}

\subsection{Optimality for heavy-tailed target distribution}

In this section we will study the optimality for heavy-tailed target distribution. The conclusion of this section is that the random-walk Metropolis algorithm has the optimality  rate $d^2$ for heavy-tailed case, which is much worse than that for the light-tailed case. 

Suppose that the target probability measure $P^d$ is in the class of (\ref{mainseceq}). 
In this case, we have an expression of $p^d(x)$ in Lemma \ref{elementary}. 
By using this, we have
\begin{align}\label{eq:mh_heavy}
	X_m^d=\left\{\begin{array}{ll}X_{m-1}^d+W_m^d&\mathrm{with\ probability}\ \alpha^d(\mathbf{S}^d_m;\|X_{m-1}^d\|^2/d)\\X_{m-1}^d&\mathrm{with\ probability}\ 1-\alpha^d(\mathbf{S}^d_m;\|X_{m-1}^d\|^2/d)
	\end{array}\right.
\end{align}
where $X_0^d\sim P^d$ and 
\begin{align*}
\alpha^d(s;z)=\frac{q^d(z(1+\frac{2s^+}{d}))}{q^d(z)}\left(1+\frac{2s^+}{d}\right)^{\frac{2-d}{2}},\ \mathbf{S}^d_m=\frac{S^d_m}{\|X_{m-1}^d\|^2/d},\ 
\mathbf{X}_m^d=\frac{X_m^d}{\sqrt{\|X_m^d\|^2/d}}
\end{align*}
Let $\mathcal{F}_M^d=\sigma\{\|X_m^d\|^2,m\le M\}\cup \{\|W_m^d\|^2,m\le M+1\}$. 
We use the following technical lemma due to \cite{2016K} for the main proposition. 

\begin{lemma}\label{lem:heavylem2}
Let $h(s)=se^{-s^+}$ and $h^d(s;z)=s\alpha^d(s;z)$. Then for $a>1$, 
\begin{equation}\nonumber
d\sup_{z\in K_a}\sup_{s\in\mathbb{R}}|h^d(s;z)-h(s))|=o(1)
\end{equation}
\end{lemma}

For the proof of the following main result, there are two important probability measures: $N_\sigma=N(\sigma^2/2,\sigma^2)$ and $\mathcal{U}_1^d$, which is the law of the first coordinate of uniformly distributed variable on $\{x\in\mathbb{R}^d; \|x\|^2=d\}$. The properties of $N_\sigma$ and $\mathcal{U}_1^d$ are summarised in Sections \ref{sec:unif} and \ref{sec:girsanov}.

\begin{proposition}\label{prop:main:heavy}
Let $P^d$ be a mixture of a normal distribution as defined in (\ref{mainseceq}). 
Let $\Gamma^d$ be probability measure on $\mathbb{R}^d$ symmetric about the origin and $\alpha(d)/d^{2-\epsilon}\rightarrow 0$ for some $\epsilon>0$. 
Then 
$\sup_{0\le s,t\le T}|Z_t^d-Z_s^d|\overset{p}{\to} 0\ (d\rightarrow\infty)$ where 
$Z^d_t=\|X^d_{[t\alpha(d)]}\|^2/d$. 
\end{proposition}

\begin{proof}
Set $\xi_m^d=\|X_m^d\|^2/d-\|X_{m-1}^d\|^2/d$. Let 
$\tau_a^d=\inf\{t\ge 0; Z_t^d\notin K_a\}$ for $a>1$. The claim will be proved 
if we can show
$\sup_{1\le M\le T^d}|\sum_{m=1}^M\xi_m^d|=o_\mathbb{P}(1)$ where $T^d=\alpha(d)\min\{T,\tau_a^d\}$. 
We apply Lemma \ref{lem:gj}. 
\begin{itemize}
	\item[(a)] We prove (\ref{lem:gi:drift}) in  Lemma \ref{lem:gj}. By stationarity, we have
	\begin{align*}
\mathbb{E}\left[\sum_{m=1}^{T^d}\left|\mathbb{E}\left[\xi_m^d|\mathcal{F}_{m-1}^d\right]\right|\right]
\le \alpha(d)T
\mathbb{E}\left[\left|\mathbb{E}\left[\xi_1^d|\mathcal{F}_{0}^d\right]\right|,Z_0^d\in K_a\right]
\end{align*}
and so the claim will be proved if the right-hand side converges to $0$. 
By (\ref{eq:mh_heavy}) together with the expression $h^d(s;z)=s\alpha^d(s;z)$, we have
	\begin{align*}
\left|\mathbb{E}\left[\xi_1^d|\mathcal{F}_{0}^d\right]\right| 1_{\{Z_0^d\in K_a\}}=
2d^{-1}\left|\mathbb{E}\left[S_1^d\alpha(\mathbf{S}_m^d;Z_0^d)|\mathcal{F}_{0}^d\right]\right| 1_{\{Z_0^d\in K_a\}}
\le 2ad^{-1}\left|\mathbb{E}\left[h^d(\mathbf{S}_1^d;Z_0^d)|\mathcal{F}_{0}^d\right]\right|1_{\{Z_0^d\in K_a\}}. 
\end{align*}
By triangular inequality, for $h(x)=xe^{-x^+}$, we have
\begin{align*}
\left|\mathbb{E}\left[h^d(\mathbf{S}_1^d;Z_0^d)|\mathcal{F}_0^d\right]\right|
&\le
	\left|
		\mathbb{E}\left[
			h^d(\mathbf{S}_1^d;Z_0^d)-h(\mathbf{S}_1^d)|\mathcal{F}_0^d
		\right]\right|
		+\left|\mathbb{E}\left[h(\mathbf{S}_1^d)|\mathcal{F}_0^d
		\right]\right|.
\end{align*}
and  by Lemma \ref{lem:heavylem2}, the first term is dominated by 
\begin{align*}
	\sup_{z\in K_a}\|h^d(\cdot;z)-h\|_\infty=O(d^{-1}).
\end{align*}
Let $\sigma^2=\|W_1^d\|^2/(\|X_0^d\|^2/d)$. 
For the second term, we will prove that for sufficiently large $d$, 
\begin{align}\label{eq:heavy:drift}
	\left|\mathbb{E}\left[h(\mathbf{S}_1^d)|\mathcal{F}_0^d
		\right]\right|
\le C\min\left\{\frac{\sigma}{d},\sigma^{-k}\right\}
\end{align}
for any $k\in\mathbb{N}$ and for some $C>0$. 
Observe that 
\begin{align*}
	\mathbb{E}\left[h(\mathbf{S}_1^d)|\mathcal{F}_0^d
		\right]=\mathbb{E}\left[h(\mathbf{S}_1^d)-N_\sigma h|\mathcal{F}_0^d
		\right]
\end{align*}
by (\ref{expectofh}) where $N_\sigma=N(\sigma^2/2,\sigma^2)$. 
The right-hand side of the above can be dominated by the Wasserstein distance 
between $\mathcal{L}(\mathbf{S}_1^d|\mathcal{F}_0^d)$ and $N_\sigma$. 
Note that  $U_1^d=\langle \mathbf{X}_0^d,\frac{W_1^d}{\|W_1^d\|}\rangle\sim\mathcal{U}_1^d$ conditioned on $\mathcal{F}_0^d$
and $\mathbf{S}_1^d=\sigma U_1^d+\sigma^2/2$.  Hence by Lemma \ref{Diaconis}
\begin{align*}
\left|\mathbb{E}\left[h(\mathbf{S}_1^d)|\mathcal{F}_0^d
		\right]-N_\sigma h\right|
		\le \|\mathcal{L}(\mathbf{S}_1^d|\mathcal{F}_0^d)-N_\sigma\|_1 
		= \sigma\|\mathcal{U}_1^d-N(0,1)\|_1\le \frac{3\sigma}{d-1}. 
\end{align*}
This proves the first half of (\ref{eq:heavy:drift}). For the latter half of (\ref{eq:heavy:drift}), we decompose 
\begin{align*}
	\left|\mathbb{E}\left[h(\mathbf{S}_1^d)|\mathcal{F}_0^d
		\right]\right|\le 
			\left|\mathbb{E}\left[h(\mathbf{S}_1^d)1_{\{\mathbf{S}_1^d>\sigma^2/4\}}|\mathcal{F}_0^d
		\right]\right|
		+
		\left|\mathbb{E}\left[h(\mathbf{S}_1^d)1_{\{\mathbf{S}_1^d\le \sigma^2/4\}}|\mathcal{F}_0^d
		\right]\right|.
\end{align*}
Since $\sup_{x\ge 0}x^{k/2}h(x)<\infty$, the first term in the right-hand side
is  $O((\sigma^2/4)^{-k/2})$. 
 The latter is also dominated by
\begin{align*}
		\left|\mathbb{E}\left[h(\mathbf{S}_1^d)1_{\{\mathbf{S}_1^d\le \sigma^2/4\}}|\mathcal{F}_0^d
		\right]\right|&\le \mathbb{E}\left[(\sigma |U_1^d|+\frac{\sigma^2}{2})1_{\{U_1^d<-\sigma/4\}}|\mathcal{F}_0^d\right]\\
		&\le 
		\mathbb{E}\left[\left.\sigma |U_1^d|\left(\frac{|U_1^d|}{\sigma/4}\right)^{k+1}+\frac{\sigma^2}{2}\left(\frac{|U_1^d|}{\sigma/4}\right)^{k+2}\right|\mathcal{F}_0^d\right]=O(\sigma^{-k})
\end{align*}
by (\ref{Stirling}). These estimates yields (\ref{eq:heavy:drift}). Since 
the right-hand side is maximised  when $\sigma=d^{1/(k+1)}$, 
we have
\begin{align*}
	\left|\mathbb{E}\left[h(\mathbf{S}_1^d)|\mathcal{F}_0^d
		\right]\right|=O(d^{-k/(k+1)})
\end{align*}
for any $k\in\mathbb{N}$. By choosing $k$ so that $\alpha(d)d^{-1-k/(k+1)}\rightarrow 0$,  the convergence holds. 
\item[(b)]
Next we prove (\ref{lem:gi:diffusion}) in Lemma \ref{lem:gj}. By simple calculation, 
\begin{align*}
	\mathbb{E}\left[\sum_{m=1}^{T^d}\mathbb{E}\left[|\xi_m^d|^2|\mathcal{F}_{m-1}^d\right]\right]&=
		\mathbb{E}\left[\sum_{m=1}^{T^d}\left(\frac{\|X_{m-1}^d\|^2}{d}\right)^2\mathbb{E}\left[\left.\left(\frac{\|X_m^d\|^2-\|X_{m-1}^d\|^2}{\|X_{m-1}^d\|^2}\right)^2\right|\mathcal{F}_{m-1}^d\right]\right]\\
		&\le
		a^2T\alpha(d)\mathbb{E}\left[\left(\frac{\|X_1^d\|^2-\|X_0^d\|^2}{\|X_0^d\|^2}\right)^2\right].
\end{align*}
By reversibility, 
\begin{align*}
\mathbb{E}\left[\left(\frac{\|X_1^d\|^2-\|X_0^d\|^2}{\|X_0^d\|^2}\right)^2\right]
&=
\mathbb{E}\left[\left\{\left(\frac{\|X_1^d\|^2-\|X_0^d\|^2}{\|X_0^d\|^2}\right)^-\right\}^2\right]
+\mathbb{E}\left[\left\{\left(\frac{\|X_1^d\|^2-\|X_0^d\|^2}{\|X_0^d\|^2}\right)^+\right\}^2\right]\\
&=
\mathbb{E}\left[\left\{\left(\frac{\|X_1^d\|^2-\|X_0^d\|^2}{\|X_0^d\|^2}\right)^-\right\}^2\right]
+\mathbb{E}\left[\left\{\left(\frac{\|X_1^d\|^2-\|X_0^d\|^2}{\|X_1^d\|^2}\right)^-\right\}^2\right]\\
&\le2\mathbb{E}\left[\left\{\left(\frac{\|X_1^d\|^2-\|X_0^d\|^2}{\|X_1^d\|^2}\right)^-\right\}^2\right]\\
&\le2\mathbb{E}\left[\left\{\left(\frac{\|X_0^d+W_1^d\|^2-\|X_0^d\|^2}{\|X_0^d+W_1^d\|^2}\right)^-\right\}^2\right]. 
\end{align*}
Recall now that $U_1^d\sim \mathcal{U}_1^d$, and hence it has $|U_1^d|^2/d\sim \mathrm{Beta}(1/2,(d-1)/2)$ by (\ref{eq:beta}).  Observe that 
\begin{align*}
	\|X_0^d+W_1^d\|^2-\|X_0^d\|^2&=2\left\langle X_0^d,W_1^d\right\rangle +\|W_1^d\|^2\\
	&= \left(\|W_1^d\|+\left\langle X_0^d,\frac{W_1^d}{\|W_1^d\|}\right\rangle\right)^2-\left|\left\langle X_0^d,\frac{W_1^d}{\|W_1^d\|}\right\rangle\right|^2\\
	&\ge -\left|\left\langle X_0^d,\frac{W_1^d}{\|W_1^d\|}\right\rangle\right|^2\\
	&=-|U_1^d|^2\left(\frac{\|X_0^d\|^2}{d}\right).
\end{align*}
Thus 
\begin{align*}
\mathbb{E}\left[\left\{\left(\frac{\|X_0^d+W_1^d\|^2-\|X_0^d\|^2}{\|X_0^d+W_1^d\|^2}\right)^-\right\}^2\right]
&\le 
\mathbb{E}\left[\left(\frac{|U_1^d|^2/d}{1-|U_1^d|^2/d}\right)^2\right]=\frac{\mathrm{B}(5/2,(d-5)/2)}{\mathrm{B}(1/2,(d-1)/2)}=O(d^{-2}). 
\end{align*}
Thus the convergence holds. 
\end{itemize}

\end{proof}

\begin{theorem}
Let $P^d$ be a mixture of a normal distribution as defined in (\ref{mainseceq}). 
	If the RWM is weakly consistent with the rate $\alpha(d)$, then $\liminf_{d\rightarrow\infty}\alpha(d)/d^{2-\epsilon}>0$ for any $\epsilon>0$. The RWM algorithm has the optimal rate $\alpha(d)=d^2$ in this sense. 
\end{theorem}

\begin{proof}
Suppose by way of contradiction that the RWM is weakly consistent with the rate $\alpha(d)$ such that $\liminf_{d\rightarrow\infty}\alpha(d)/d^{2-\epsilon}=0$. Without loss of generality, by taking subsequence, we can assume $\alpha(d)=o(d^{2-\epsilon})$. 
Let $M(d)=d^{2-\epsilon}$. By Proposition \ref{prop:main:heavy}, 
\begin{equation}\label{eq:main:theorem}
\mathcal{L}\left(\frac{1}{M(d)}\sum_{m=0}^{M(d)-1}\frac{\|X^d_m\|^2}{d}\right)\rightarrow Q. 
\end{equation}
\begin{enumerate}
\item[(a)] Suppose that $\int yQ(\dif y)<\infty$. 
By construction (\ref{mainseceq}), we
have $\int yQ(\dif y)=\int yQ_1(\dif y)$. Since $(X_{m,k}^d)^2\sim Q_1$, by weak consistency, 
\begin{align*}
\mathbb{E}\left[\left|\frac{1}{M(d)}\sum_{m=0}^{M(d)-1}\left(X_{m,k_d}^d\right)^2-\int yQ(\dif y)\right|\right]=o(1)
\end{align*} 
for any choice of $k_d\in\{1,\ldots, d\}$. 
Hence
\begin{align*}
\mathbb{E}\left[\left|\frac{1}{M(d)}\sum_{m=0}^{M(d)-1}\frac{\|X_m^d\|^2}{d}-\int yQ(\dif y)\right|\right]
&\le d^{-1}\sum_{k=1}^d\mathbb{E}\left[\left|\frac{1}{M(d)}\sum_{m=0}^{M(d)-1}(X_{m,k}^d)^2-\int yQ(\dif y)\right|\right]\\
&\le \mathbb{E}\left[\left|\frac{1}{M(d)}\sum_{m=0}^{M(d)-1}(X_{m,k_d}^d)^2-\int yQ(\dif y)\right|\right]=o(1)
\end{align*}
where we choose $k_d\in\{1,\ldots, k\}$ which maximises the expectation. 
By (\ref{eq:main:theorem}), this implies  $Q=\int y Q(\dif y)$ a.s. This is impossible since $Q$ has a probability density. 
\item[(b)] Suppose that $\int yQ(\dif y)=\int yQ_1(\dif y)=+\infty$. Then as in the case (a), it is straightforward to show
\begin{align*}
\frac{1}{M(d)}\sum_{m=0}^{M(d)-1}\frac{\|X_m^d\|^2}{d}
=&
\frac{1}{dM(d)}\sum_{m=0}^{M(d)-1}\sum_{k=1}^d|X_{m,k}^d|^2\\
\ge & 
\frac{1}{dM(d)}\sum_{m=0}^{M(d)-1}\sum_{k=1}^d\min\{|X_{m,k}^d|^2,K\}\overset{p}{\to} \int\min\{y,K\}Q_1(\dif y)
\end{align*}
for each $K>0$ under consistency assumption. 
This implies
\begin{align*}
\frac{1}{M(d)}\sum_{m=0}^{M(d)-1}\frac{\|X_m^d\|^2}{d}\overset{p}{\to} +\infty. 
\end{align*}
It contradicts (\ref{eq:main:theorem}) since the left-hand side converges in law to $Q$. 
\end{enumerate}
Thus for each case, the RWM algorithm does not have the rate $\alpha(d)$ such that $\liminf_{d\rightarrow\infty}\alpha(d)/d^{2-\epsilon}=0$. 
On the other hand, by \cite{2016K}, the RWM algorithm has weak consistency with the rate $d^2$
when $\Gamma^d=N_d(0,l^2I_d/d)$ for $l>0$. Thus $\alpha(d)=d^2$ is the optimal rate in this sense. 
\end{proof}

%
%

\appendix

\section{Limit theory}

In this paper, the proof for degenerate limit uses Lemma 9 of \cite{MR1204521}. 
Let $(\Omega^d,\mathcal{F}^d,\mathbf{F}^d=(\mathcal{F}_m^d)_m,\mathbb{P}^d)$ be a discrete-time stochastic base. 

\begin{lemma}[Lemma 9 of \cite{MR1204521}]\label{lem:gj}
Let $\xi_m^d$ be $\mathbb{R}$-valued $\mathcal{F}_m^d$-measurable and let $T^d$ be a  stopping time for each $d$. The following two conditions imply 
$\sup_{1\le M\le T^d}\left|\sum_{m=1}^M\xi_m^d\right|\overset{p}{\to}0$:
	\begin{align}
		&\sum_{m=1}^{T^d}\left|\mathbb{E}[\xi_m^d|\mathcal{F}_{m-1}^d]\right|\overset{p}{\to}0,\label{lem:gi:drift}\\
		&\sum_{m=1}^{T^d}\mathbb{E}[|\xi_m^d|^2|\mathcal{F}_{m-1}^d]\overset{p}{\to}0.\label{lem:gi:diffusion}
	\end{align}
\end{lemma}

\begin{proof}
Let $\eta_m^d=\xi_m^d-\mathbb{E}[\xi_m^d|\mathcal{F}_{m-1}^d]$. Set 
\begin{align*}
	B_t^d:=\sup_{1\le M\le t}\left|\sum_{m=1}^M\eta_m^d\right|,\ 
	C_t^d:=\sum_{1\le m\le t}\mathbb{E}[|\eta_m^d|^2|\mathcal{F}_{m-1}^d],\ 
	D_t^d:=\sum_{1\le m\le t}\mathbb{E}[|\xi_m^d|^2|\mathcal{F}_{m-1}^d].
\end{align*}
It is sufficient to prove $B_{T^d}^d\overset{p}{\to}0$. 
Since 
	\begin{align*}
		\left\{\sum_{1\le m\le \cdot}\eta_m^d\right\}^2-\sum_{1\le m\le \cdot}\mathbb{E}[(\eta_m^d)^2|\mathcal{F}_{m-1}^d]
	\end{align*}
	is a $\mathbf{F}^d$-local martingale, we can apply Lenglart inequality Lemma I.3.30 of \cite{JS} for any $\epsilon, \eta>0$, and 
	\begin{align*}
		\mathbb{P}\left(|B_{T^d}^d|^2\ge \epsilon\right)
		\le \frac{\eta}{\epsilon}+\mathbb{P}\left(C_{T^d}^d\ge \eta\right).
	\end{align*}
	Hence the claim follows by $C_t^d\le D_t^d$. 
\end{proof}

\begin{corollary}\label{cor:gj}
Let $\xi_m^d=(\xi_{m,1}^d,\ldots, \xi_{m,d}^d)\in\mathbb{R}^d$ be $\mathcal{F}_m^d$-measurable stationary process and let $\alpha(d)/d\rightarrow 0$. The following two conditions imply 
$\sup_{1\le M\le \alpha(d)t}\left|\sum_{m=1}^M\xi_{m,k_d}^d\right|\overset{p}{\to}0$
for some $k_d\in\{1,\ldots, k_d\}$:
	\begin{align}\label{eq:cor:gi}
		\sum_{k=1}^{d}\mathbb{E}\left[\left|\mathbb{E}[\xi_{1,k}^d|\mathcal{F}_0^d]\right|\right]<\infty,\ \mathbb{E}[\|\xi_1^d\|^2]<\infty.
	\end{align}
\end{corollary}

\begin{proof}
	By assumption, there exists $k_d$ such that 
	\begin{align*}
		\mathbb{E}\left[\left|\mathbb{E}[\xi_{1,k_d}^d|\mathcal{F}_0^d]\right|\right]=O(d^{-1}),\ \mathbb{E}[|\xi_{1,k_d}^d|^2]=O(d^{-1}).
	\end{align*}
	Thus the claim follows by the previous lemma since $(\xi_m^d)_m$ is stationary. 
\end{proof}

\section{Exchangeability and Stein methods}

\subsection{Exchangeable pair}

In this paper, a coupled random variables $(X, Y)$ is called exchangeable pair
if $\mathcal{L}(X,Y)=\mathcal{L}(Y,X)$. 

\begin{lemma}\label{lem:exchageble}
If $(X,Y)$ is an exchangeable pair, for any bounded function $f(x)$, we have 
		\begin{align}
	&\mathbb{E}[f(X)]=\frac{1}{2}\mathbb{E}[f(X)+f(Y)]\label{eq:lem:exchageble1},\\
	&\frac{1}{2}\mathbb{E}[|f(X)-f(Y)|]=2\mathbb{E}\left[f(X)\left(\frac{1}{2}-\mathbb{P}(f(X)<f(Y)|X)\right)\right]\label{eq:lem:exchageble2}
\end{align}
\end{lemma}
\begin{proof}
The former equation is obvious. The latter is 
\begin{align*}
	\frac{1}{2}\mathbb{E}[|f(X)-f(Y)|]=&\mathbb{E}[f(X)-f(Y),f(X)>f(Y)]\\
	=&\mathbb{E}[f(X)]-\mathbb{E}[f(X),f(X)<f(Y)]
	-\mathbb{E}[f(Y),f(X)>f(Y)]\\
	=&\mathbb{E}[f(X)]-2\mathbb{E}[f(X),f(X)<f(Y)]\\
	=&\mathbb{E}\left[f(X)\left(1-21_{\{f(X)<f(Y)\}}\right)\right]\\
	=&\mathbb{E}\left[f(X)\left(1-2\mathbb{P}(f(X)<f(Y)|X)\right)\right]. 
\end{align*}
\end{proof}

\subsection{Uniform distribution on the unit sphere}\label{sec:unif}

Let $U^d=(U^d_i)_{i=1,\ldots, d}$ be uniformly distributed  on the  sphere $\{x\in\mathbb{R}^d;\|x\|^2=d\}$. 
We state some properties of the law of $U_1^d$, denoted by $\mathcal{U}_1^d$. If $v\in\mathbb{R}^d$ has $\|v\|=1$, then $\langle U^d, v\rangle \sim \mathcal{U}_1^d$. In particular, 
$\langle U^d, \frac{V}{\|V\|}\rangle\sim \mathcal{U}_1^d$ for $\mathbb{R}^d$-valued random variable $V$ which is independent of $U^d$. 
Also if $X^d=(X^d_i)_{i=1,\ldots, d}\sim N_d(0,I_d)$, then $d^{1/2}X^d_1/\|X^d\|\sim \mathcal{U}_1^d$. Hence
\begin{align}\label{eq:beta}
	\mathcal{L}\left(\frac{\|U_1^d\|^2}{d}\right)=\mathcal{L}\left(\frac{|X_1^d|^2}{\sum_{i=1}^d|X_i^d|^2}\right)\sim \mathrm{Beta}(1/2,(d-1)/2)
\end{align}
since each $|X_i^d|^2$ follows the chi-squared distribution. 
For each $d$, 
$U_1^d$ has the  mean $0$ and the variance $1$
and as $d\rightarrow\infty$, it tends to the standard normal distribution.  Asymptotic normality result with a sharp bound can be found in p399 of \citet{MR898502}
for the total variation distance and Theorem 10 of \citet{MR2453473} for the Wasserstein distance.

\begin{lemma}\label{Diaconis}
\begin{equation}\nonumber
\|\mathcal{U}_1^d-N(0,1)\|_{\mathrm{TV}}\le \frac{8}{d-4},\ 
\|\mathcal{U}_1^d-N(0,1)\|_{1}\le \frac{3}{d-1}
\end{equation}
where $\|P-Q\|_{\mathrm{TV}}=2\sup_A |P(A)-Q(A)|$
and $\|P-Q\|_{1}=\sup_{f\in B_1} |P(f)-Q(f)|$, where $B_1$ is the set of functions $f$
such that $\sup|f(x)-f(y)|/|x-y|\le 1$. 
\end{lemma}

In addition to Lemma \ref{Diaconis}, we will use
\begin{equation}\label{Stirling}
\mathbb{E}\left[(U_1^d)^\alpha\right]=\frac{d^{\alpha/2}B(\frac{\alpha+1}{2},\frac{d-1}{2})}{B(\frac{1}{2},\frac{d-1}{2})}\rightarrow  \frac{\Gamma(\frac{\alpha+1}{2})}{\Gamma(\frac{1}{2})}2^{\alpha/2}\ (d\rightarrow\infty)
\end{equation}
for $\alpha>-1$, where we used Stirling's approximation.

\subsection{Girsanov's formula}\label{sec:girsanov}
In this paper, we focus on $N_\sigma=N(\sigma^2/2,\sigma^2)$. 
An important property for $N_\sigma$ is the following Girsanov's formula:
if $S\sim N_\sigma$, then for any bounded measurable function $f$, 
\begin{equation}\label{Girsanov}
\mathbb{E}\left[f(S),S<0\right]=
\mathbb{E}\left[f(-S)e^{-S},S>0\right]. 
\end{equation}
This formula is used throughout in this paper. 
In particular,  by taking $f(s)=z$, 
\begin{equation}\label{expectofh}
N_\sigma h=\mathbb{E}\left[h(S)\right]=0
\end{equation}
 for $h(s)=se^{-s^+}$, 
 and by taking $f(s)=1$, 
 \begin{equation}\nonumber
 \mathbb{E}[\exp(-S),S>0]=\mathbb{P}(S<0)=\mathbb{E}[\exp(-S^{+})]/2.
 \end{equation} 

\noindent
For $k\ge 0$, $\mu_k(\sigma):=\mathbb{E}\left[|S|^ke^{-S^+}\right]=2(-1)^k\mathbb{E}\left[S^k,S<0\right]=2\mathbb{E}\left[S^ke^{-S},S>0\right]$ 
for $S\sim N_\sigma$. It is not difficult to conclude  
\begin{equation}\nonumber
\mu_0(\sigma)=2\Phi\left(-\frac{\sigma}{2}\right),\ \mu_1(\sigma)=-\sigma^2\Phi\left(-\frac{\sigma}{2}\right)+2\sigma\phi\left(-\frac{\sigma}{2}\right), 
\end{equation}
In particular, $\sup_{\sigma>0}\sigma^k\mu_l(\sigma)<\infty$ for any $k\in\mathbb{N}_0, l=0,1$.

\bibliographystyle{plainnat}

\end{document}